\documentclass{article}
\usepackage{amssymb}
\usepackage{amsmath,amsthm}
\usepackage{setspace}
\usepackage{tabularx}
\usepackage{booktabs}
\usepackage{longtable}
\usepackage{graphicx}
\usepackage{adjustbox}
\usepackage{soul,color}
\usepackage{lipsum}
\usepackage{inputenc,fullpage}
\usepackage [english]{babel}
\usepackage [autostyle, english = american]{csquotes}
 \usepackage{verbatim}
 \usepackage{comment}
 \usepackage{longtable}
\usepackage{array,arydshln,tabu}
\usepackage{cite}
\MakeOuterQuote{"}
\usepackage{xcolor,soul,verbatim}
\newcommand\aug{\fboxsep=-\fboxrule\!\!\!\fbox{\strut}\!\!\!}
\begin{document}

\newtheorem{theorem}{Theorem}[section]
\newtheorem{definition}{Definition}[theorem]
\newtheorem{corollary}{Corollary}[theorem]
\newtheorem{lemma}[theorem]{Lemma}
\newtheorem{prop}{Proposition}
\newtheorem{remark}{Remark}
\newtheorem{note}{Note}
\newtheorem{conjecture}[theorem]{Conjecture}
\newcommand{\hlc}[2][yellow]{{%
    \colorlet{foo}{#1}%
    \sethlcolor{foo}\hl{#2}}%
}
\newcommand{\Fq}{\mathbb{F}_q}
\newcommand{\Z}{\mathbb{Z}}
\newcommand{\qrn}{{\mathbb{F}_q[x] \over \langle x^n-1\rangle}}
\newcommand{\qrm}{{\mathbb{F}_q[x] \over \langle x^m-a\rangle}}
\addto\captionsenglish{\renewcommand{\tablename}{}}

\addto\captionsenglish{\renewcommand{\tablename}{}}
%\begin{document}

%\pagenumbering{gobble}
\title{Good Classical and Quantum Codes from Multi-Twisted Codes}
\author{Nuh Aydin, Thomas Guidotti, and Peihan Liu } 
\date{May 2020}

\maketitle
\begin{abstract}
Multi-twisted (MT) codes were introduced as a generalization of quasi-twisted (QT) codes. QT codes have been known to contain many good codes.  In this work, we show that codes with good parameters and desirable properties can be obtained from MT codes. These include best known and optimal classical codes with additional properties such as reversibility and self-duality, and new and best known non-binary quantum codes obtained from special cases MT codes. Often times best known quantum codes in the literature are obtained indirectly by considering extension rings. Our constructions have the advantage that we obtain these codes by more direct and simpler methods. Additionally, we found theoretical results about binomials over finite fields that are useful in our search. 


\end{abstract}
\footnote{\textit{Email addresses:} \begin{scriptsize}\textbf{aydinn@kenyon.edu}\end{scriptsize} (Nuh Aydin), \begin{scriptsize}\textbf{guidotti1@kenyon.edu}\end{scriptsize} (Thomas Guidotti), \begin{scriptsize}\textbf{paulliu@umich.edu, paulpeihanliu@gmail.com}\end{scriptsize} (Peihan Liu)}

\section{Introduction and Motivation}
A liner code $C$ of length $n$ over $\Fq$, the finite field of order $q$, is a vector subspace of $\Fq^n$. The three important parameters of a linear code are: the length ($n$),  the dimension ($k$), and the minimum distance ($d$). If a code $C$ over $\Fq$ is of length $n$, dimension $k$, and minimum distance $d$, then we call it an $[n,k,d]_q$ code. One of the most important and challenging problems in coding theory is to determine the optimal parameters of a linear code and to explicitly construct codes that attain them. For example, given the alphabet size $q$, length $n$, and dimension $k$, researchers try to determine the highest value of the minimum distance $d$, and  give an explicit construction of such a code. Despite much work on this question, it still has many open cases.  There are databases that contain information about best known linear codes (BKLC) and their constructions. Two of the well known databases  are the maintained by M. Grassl ( \cite{otable}, codetables.de)  and the database contained in the Magma software \cite{magma}.

Given $q, n$ and $k$, there are theoretical upper bounds on $d$ that are given in the tables. However,   it is not guaranteed and not always possible that they can actually be attained.  This is a challenging problem for two main reasons: firstly, computing the minimum distance of a linear code is an NP-hard problem \cite{NPhard} and secondly the number $\displaystyle{\frac{(q^n-1)(q^n-q)\cdots (q^n-q^{k-1})}{(q^k-1)(q^k-q)\cdots (q^k-q^{k-1})}}$ of linear codes  over $\Fq$ of length $n$ and dimension $k$ grows fast when $n,k$ and $q$ increase. Therefore, it is impossible to do an exhaustive search except for very small parameters. Hence, researchers focus  work on sub-classes of linear codes  with special algebraic properties. One very promising class has been the class of quasi-twisted (QT) codes which contains the important classes such as cyclic codes, constacyclic codes, and quasi-cyclic (QC) codes as special cases. Hundreds of BKLCs have been obtained, usually by computer searches, from the class of QT codes. 

More recently, a generalization of QT codes called multi-twisted (MT) codes was introduced in \cite{mt}. Therefore, MT codes are a generalization of cyclic, negacyclic, consta-cyclic (CC), quasi-cyclic (QC), generalized quasi-cyclic (GQC) and quasi-twisted codes (QT). Compared to QT codes, the search space for MT codes is much larger. While this increases the chances of finding new codes, it also makes the search process computationally more expensive. In this work, we focus on a special type of MT codes whose generator matrix is of the form $G= [I_k \mid A]$ where $I_k$ is the $k \times k$ identity matrix and $A$ is a $k \times n-k$ circulant matrix. The  idea of investigating  these codes was first introduced in \cite{ICY}. In addition to trying to find new linear codes from this type of MT Codes, we also worked on obtaining codes that are reversible, self-dual, and self-orthogonal.

Finally, we also worked on obtaining quantum codes from classical codes, specifically from constacyclic, QC, QT, and MT codes. The idea of quantum codes was first introduced in \cite{Quantumoriginal1} and \cite{Quantumoriginal2}. Later methods of constructing quantum codes from classical codes were introduced in \cite{Quantumoriginal3} (known as CSS construction)  and \cite{steane} (Steane construction). Since then, researchers have been working on various ways of obtaining quantum codes from classical codes. More recently, codes over extension rings have attracted a lot of attention. Many quantum codes in the literature have been obtained via the method of considering codes over an extension ring $S$ of a base/ground ring (or field) $R$ and using a map to eventually obtain codes over the ground ring $R$. We show that in many cases, there is no need to obtain these codes in such indirect way. We have obtained many of these best known codes presented in the literature, and even better ones, by using a direct approach. It is preferable to obtain codes in more direct ways with good algebraic structures. 

%to construct new . One of the famous construction is CSS construction that was first introduced in \cite{Quantumoriginal1} and \cite{Quantumoriginal2}. By CSS construction, we found a lot of new quantum codes.

\section{Preliminaries} 
\begin{definition} \cite{mt} For each $i=1,\dots,\ell$, let $m_i$ be a positive integer and $a_i\in \Fq ^{*}=\Fq\setminus \{0\}$. A multi-twisted (MT) module $V$ is an $\Fq[x]$-module of the form 
\[V = {\displaystyle\prod_{i=1}^\ell \Fq[x] / \langle x^{m_i} - a_i \rangle},\]  An MT code is an $\Fq[x]$-submodule of an MT module $V$. 
\end{definition}

Equivalently, we can define an MT code in terms of the shift of a codeword. Namely, a linear code $C$ is MT if for any codeword 
$$\vec{c}=(c_{1,0},\dots, c_{1,{m_1-1}}; c_{2,0},\dots, c_{2,{m_2-1}};\dots;c_{\ell,0},\dots, c_{\ell,{m_{\ell}-1}})\in C,$$
\noindent its multi-twisted shift
$$(a_1c_{1,m_1-1},c_{1,0},\dots, c_{1,{m_1-2}}; a_2c_{2,m_2-1},c_{2,0},\dots, c_{2,{m_2-2}};\dots; a_{\ell}c_{\ell,m_{\ell}-1},\dots, c_{\ell,{m_{\ell}-2}})$$
\noindent is also a codeword. If we identify a vector $\vec{c}$ with $C(x)=(c_1(x),c_2(x),\dots,c_{\ell}(x))$ where $c_i(x)=c_{i,0}+c_{i,1}x+\cdots+c_{i,m_i-1}x^{m_i-1}$, then the MT shift corresponds to $xC(x)=(xc_1(x) \mod x^{m_1}-a_1, \dots, xc_{\ell}(x) \mod x^{m_{\ell}}-a_{\ell})$.

The following are some of the most important special cases of MT codes:
\begin{itemize}
\item $a_1=a_2=\dots=a_{\ell}$ gives QT codes
\item $a_1=a_2=\dots=a_{\ell}=1$ gives QC codes
\item $\ell=1, a_1=1$  gives cyclic codes
\item $\ell=1, a_1=-1 $  gives negacyclic codes
\item $\ell=1$  gives constacyclic codes
\end{itemize}

There has been much research on  QT codes, a special case of MT codes. Hundreds of new codes have been obtained from QT codes by computer searches. An algorithm called ASR has been particularly effective in this regard.The ASR algorithm is based on the following theorem.

\begin{theorem}\label{theorem:ASR}
\cite{qtmain} Let C be a 1-generator QT code of length $n = m\ell$ over $\Fq$
with a generator of the form
$$(g(x)f_1(x), g(x)f_2(x), . . . , g(x)f_{\ell}(x))$$ where $x^m - a = g(x)h(x)$ and $\gcd(h(x),f_i(x)) = 1$ for all $i = 1,...,\ell$. Then $dim(C) = m -deg(g(x))$, and $d(C) \geq \ell\cdot d$ where $d$ is the minimum distance of the constacyclic code $C_g$ generated by $g(x)$.
\end{theorem}
This algorithm has been refined and automatized  in more recent works such as (\cite{gf2}, \cite{gf7}, \cite{gf3}, \cite{nonprime}, \cite{oliverequiv}) and dozens of record breaking codes have been obtained over every finite field $\Fq$, for $q=2,3,4,5,7,8,9$ through its implementation. Moreover, it has been further generalized in \cite{ICY} and more new codes were discovered  that would have been missed by its earlier versions.

We can consider MT analogues of the QT codes given in the previous theorem. The following theorem gives a fact about duals of such codes.

\begin{theorem}
Let an  $C=\langle g_1(x),g_2(x),g_3(x),...,g_{\ell}(x) \rangle$ be an MT code such that  the constacyclic codes generated by $g_i(x)$ have the same length $n$ and the same dimension $k$ for $i=1,2,...,\ell$. Let $D_i$ be the dual of $\langle g_i(x) \rangle$. Let $PD=D_1\times D_2\times \dots \times D_{\ell}$. Then $PD$ is contained in $C^{\perp}$. 
\end{theorem}
\begin{proof}
Let $G_i$ be a generator matrix of $C_i=\langle g_i(x)\rangle$, so a generator matrix $G$ of $C$ is given by
\begin{align*}
    G=
    \begin{pmatrix}
    G_1&G_2&G_3&...&G_l
    \end{pmatrix}_{k\times \sum_{i=1}^{l} n_i}.
\end{align*}
Let $H_i$ be a generator matrix of $D_i=C_i^{\perp}$. Then a generator matrix $GPD$ of $PD$ is given by
\begin{align*}
    GPD=
    \begin{pmatrix}
    H_1&0&0&...&0\\
    0&H_2&0&...&0\\
    0&0&H_3&...&0\\
    \ddots&\ddots&\ddots&\ddots&\ddots\\
    0&0&0&...&H_l
    \end{pmatrix}_{(n-k)l \times \sum_{i=1}^{\ell} n_i}.
\end{align*}
It is readily verified  that $G\cdot GPD=0$. Note that the dimension of $C^{\perp}$ is $\sum_{i=1}^{\ell} n_i-k$ and the dimension of $PD$ is $\sum_{i=1}^{\ell}(n_i-k)$, so the dimension of $PD$ is $\leq$ the dimension of $C^{\perp}$. Hence, $PD\subseteq C^{\perp}$.
\end{proof}

\section{A Special Type of MT Codes}
We first introduced a special type of MT codes that we called  ICY codes in  \cite{ICY}. The basic idea for this construction is inspired by the well known fact that any linear code is equivalent to a linear code whose generator matrix is of the form $G= [I_k \mid A]$ where $I_k$ is the $k \times k$ identity matrix and $A$ is a $k \times n-k$ matrix. A code with a generator matrix in this form is called a systematic code. In this search method, we begin with a $k \times k$ identity matrix on the left  and append it with $k \times n-k$ circulant matrices on the right. The circulant matrix is constructed using a generator polynomial $g(x) | x^{n-k} - a$ of a constacyclic code (obtained using the partition program from \cite{oliverequiv}). For each polynomial $g(x)$, we consider a number of circulant matrices obtained from polynomials $g(x) \cdot f(x)$ where $\gcd(f(x), h(x)) = 1$, with the goal of finding a code with minimum distance equal to or greater than the BKLCs for parameters $[n, k]_q$.  In \cite{ICY} we were able to find many good codes over $GF(2)$ and $GF(5)$. In this paper we present a number of optimal codes over $GF(2)$ with parameters  $[2\cdot k, k]$  obtained from an implementation of this method. A number of these codes we found are also self-dual, which  makes these codes even more useful. 
\begin{lemma} \label{theorem:ICY 2kk GF(2)}
If f(x) is a monomial and the block length of $\langle f(x)\rangle$ is the same as the block length of the identituy matrix, then the MT code $C$ over GF(2) generated by $\langle 1,f(x)\rangle$ is always self-dual, reversible and of minimum distance 2. In particular, if the monomial is a constant, then $C$ is a cyclic code.
\end{lemma}
\begin{proof}
Let $f(x)=x^i$ be a monomial and the block length of $\langle f(x)\rangle$ be the same as $k\times k$ identity matrix. Then the generator matrix of the MT code $C$ over $GF(2)$ generated by $\langle 1,f(x)\rangle$ is given by:
\begin{align*}
G=
\begin{pmatrix}
1&0&...&0&\aug&0&...&1&0&...\\
0&1&...&0&\aug&0&...&0&1&...\\
\ddots&\ddots&\ddots&\ddots&\aug&\ddots&\ddots&\ddots&\ddots&\ddots\\
0&0&...&1&\aug&...&1&0&0&...\\
\end{pmatrix}_{k \times 2k}.
\end{align*}

It is easy to verify that the reverse of the $j^th$ row is the $(k-i+1-j)^th$ row, so it follows that $C$ is reversible. Also note that each row of $G$ is orthogonal to itself and also orthogonal to every other row, so given $k=2k/2$ we know $C$ is also self-dual. Moreover, since each row is of weight 2, and both matrices are circulant, it follows that $C$ is of minimum distance 2.

Specifically, if $f(x)=1$, then the generator matrix $G$ is given by
\begin{align*}
G=
\begin{pmatrix}
1&0&...&0&\aug&1&...&0&...\\
0&1&...&0&\aug&0&1&0&...\\
\ddots&\ddots&\ddots&\ddots&\aug&\ddots&\ddots&\ddots&\ddots\\
0&0&...&1&\aug&...&0&...&1\\
\end{pmatrix}_{k \times 2k}.
\end{align*}
Hence, the code is a cyclic code of length $2k$ generated by $f(x)=x^k-1|x^{2k}-1$.
\end{proof}

The following codes in Table 1 are obtained by ICY method where we chose $n=2k$, so each one of these codes has parameters $[2k,k,d]$. All of these codes are optimal with additional properties of being reversible and/or self orthogonal. For the circulant matrix on the right, we used a random $m \times m$ circulant matrix  by choosing a random vector and taking its cyclic shifts. This polynomial is given in the table. We also checked condition for reversibility and self-duality. Note that self-orthogonality is the same as self-duality since we set $n=2k$. In this search, the binary field fields turned out to be  most promising.

In this and every subsequent table presented in this paper the terms of the generator polynomials are listed in ascending order from left to right. For instance, the first generator polynomial $[111]$ in this table is the vector representation of $g(x) = 1 + x + x^2$. 

\begin{small}
\begin{singlespacing}
\label{tab:Table1}  
% Give a unique label
% For LaTeX tables use
\begin{longtable}{ p{2cm}p{4.6cm} }
\caption*{Table 1: Optimal codes from ICY method with several properties. Here, $^*$ delineates a reversible code, and $^\circ$ delineates a self-dual code.}

\\
\hline\noalign{\smallskip}
$[n,k,d]_{q}$   & Polynomials  \\
\noalign{\smallskip}\hline\noalign{\smallskip}
    $[8, 4, 4]_2$ $^{* \circ}$ & $g=[111]$\\
    $[10, 5, 4]_2$&$g=[10101]$\\
    $[12,6,4]_{2}$ $^{* \circ}$   &$g=[110111]$\\
    $[14,7,4]_{2}$  &$g= [111101]$ \\
    $[20,10,6]_{2}$   &$g= [11111011]$ \\
    $[24, 12, 8]_2$ $^{* \circ}$ & $g=[101111011]$\\
    $[28,14,8]_{2}$   &$g= [11101111001001]$ \\
    $[32, 16, 8]_2$ $^{* \circ}$ & $g=[1100100001100011]$\\
    $[36, 18, 8]_2$ $^{* \circ}$ & $g=[111010010111001]$\\
    $[42,21,10]_{2}$   &$g= [11110011011011101001]$ \\
    $[48, 24, 12]_{2}$ $^{* \circ}$   &$g= [10111101001011110011101]$ \\
    $[16,8,6]_{3}$ $^*$   &$g= [212222]$ \\
    $[18,9,6]_{3}$ $^*$  &$g= [201002222]$ \\
    $[20,10,7]_{3}$ $^*$   &$g= [112012102]$ \\
    $[30,15,9]_{3}$ $^*$   &$g= [201212222200121]$ \\
    $[14,7,6]_{4}$ $^{* \circ}$   &$g= [\alpha\alpha\alpha^2\alpha\alpha^2\alpha^2]$ \\
    $[20,10,8]_{4}$ $^{* \circ}$   &$g=[\alpha\alpha1\alpha00\alpha^2 1\alpha^2\alpha^2]$ \\
    $[24,12,9]_{4}$  &$g=[10\alpha\alpha^2\alpha\alpha^2 \alpha10\alpha1]$ \\
    $[14,7,6]_{5}$   &$g= [3343442]$ \\
    $[16,8,7]_{5}$   &$g= [34404411]$ \\
    $[12,6,6]_{7}$   &$g= [551641]$ \\
    $[14,7,7]_{7}$   &$g= [5226211]$ \\
    $[16,8,7]_{7}$   &$g= [42325344]$ \\
    $[18,9,8]_{7}$   &$g= [524635101]$ \\
  [.5ex]
\noalign{\smallskip}\hline
\end{longtable}
 %In this table, $^*$ delineates a reversible code, and $^\circ$ delineates a self-dual code.
\end{singlespacing}
\end{small}
%\noindent In this table, $^*$ delineates a reversible code, and $^\circ$ delineates a self-dual code.

The next table presents a list of optimal self-dual codes together with their weight enumerators obtained by this method. For codes over $GF(4)=\{0,1,\alpha,\alpha+1 \}$, $\alpha$ is a root of $x^2+x+1$ over the binary field. Again, $^*$ delineates a reversible code. 
\begin{singlespacing}
\begin{footnotesize}
\label{tab:2}       % Give a unique label
% For LaTeX tables use
\begin{longtable}{ p{1.8cm}p{1.3cm}p{4.5cm}p{4.48cm} }
\caption*{Table 2: Optimal self-dual code from ICY method. }

\\
\hline\noalign{\smallskip}
$[n,k,d]_{q}$   &Type& weight enumerator & Polynomials  \\
\noalign{\smallskip}\hline\noalign{\smallskip}
    $[8, 4, 4]_2$ $^{*}$ &II&\scriptsize $A^8 + 14A^4 + 1$ & $g=[111]$\\
     $[12,6,4]_{2}$  & $I$&\scriptsize $A^{12} + 15A^8 + 32A^6 + 15A^4 + 1$ &$g=[110111]$\\
  [.5ex] 
   $[16,8,4]_{2}$ $^*$  &$II$&\scriptsize $A^{16} + 28A^{12} + 198A^8 + 28A^4 + 1$ &$g=[110111]$\\
  [.5ex]
   $[16,8,4]_{2}$ $^*$ &$I$&\scriptsize $A^{16} + 12A^{12} + 64A^{10} + 102A^8 + 64A^6 + 12A^4 + 1$ &$g=[1010111]$\\
  [.5ex]
  $[18, 9, 4]_2$ $^*$  &$I$& $\scriptsize A^{18} + 9A^{14} + 75A^{12} + 171A^{10} + 171A^8 + 75A^6 + 9A^4 + 1$ & $g=[11011001]$\\
   $[20,10,4]_{2}$ $^*$  &$I$&\scriptsize $A^{20} + 45A^{16} + 210A^{12} + 512A^{10} + 210A^8 + 45A^4 + 1$ &$g=[1011111111]$\\
  [.5ex]
  $[22,11,6]_{2}$ $^*$  &$I$&\scriptsize $ A^{22} + 77A^{16} + 330A^{14} + 616A^{12} + 616A^{10} + 330A^8 + 77A^6+ 1$ &$g=[1001011]$\\
  [.5ex]
   $[26,13,6]_{2}$   &$I$& \scriptsize$A^{26} + 52A^{20} + 390A^{18} + 1313A^{16} + 2340A^{14} + 2340A^{12} + 1313A^{10} + 390A^8 + 52A^6 + 1$ &$g=[10111110111]$\\
  [.5ex]
     $[52,26,10]_{2}$  &$I$& \scriptsize$A^{52} + 442A^{42} + 6188A^{40} + 53040A^{38} + 308958A^{36} + 1270360A^{34} + 3754569A^{32} + 8065616A^{30} + 12707500A^{28} + 14775516A^{26} +12707500A^{24} + 8065616A^{22} + 3754569A^{20} + 1270360A^{18} + 308958A^{16} + 53040A^{14} + 6188A^{12} + 442A^{10} + 1$ &$g=[1001111100101010100100101]$\\
  $[14,7,6]_{4}$ $^*$ & $Euclidean$  &\scriptsize $318A^{14} + 1302A^{13} + 2940A^{12} + 3990A^{11} + 3738A^{10} + 2226A^{9}+ 1155A^8 + 546A^7 + 168A^6 + 1$ &$g=[\alpha1\alpha\alpha^2\alpha\alpha\alpha^2]$\\
  [.5ex]
  $[16,8,6]_{4}$ $^*$ &$Euclidean$ &\scriptsize $681A^{16} + 3696A^{15} + 7896A^{14} + 14112A^{13} + 16464A^{12} + 9408A^{11} + 7392A^{10} + 4704A^9 + 678A^8 + 336A^7 + 168A^6 + 1$ &$g=[\alpha\alpha1\alpha^2\alpha\alpha^2]$\\
  [.5ex]
  $[20,10,8]_{4}$ $^*$ &$Euclidean$ &\scriptsize $171A^{12} + 432A^{11} + 864A^{10} + 1440A^9 + 459A^8 + 432A^7 + 288A^6 + 9A^4 + 1$ &$g=[\alpha\alpha1\alpha00\alpha^2\alpha^2\alpha^2]$\\
  $[22,11,8]_{4}$ $^*$ &$Euclidean$ &\scriptsize $171A^{12} + 432A^{11} + 864A^{10} + 1440A^9 + 459A^8 + 432A^7 + 288A^6 + 9A^4 + 1$ &$g=[\alpha^2\alpha\alpha\alpha^2010\alpha \alpha^2\alpha^2\alpha]$\\
  [.5ex] 
\noalign{\smallskip}\hline
\end{longtable}
\end{footnotesize}
\end{singlespacing}

%In this table, $^*$ delineates reversible code and $^\circ$ delineates self-dual codes.

\section{Reversible Codes}
Reversible codes are essential in application of coding theory to DNA computing \cite{dna2}. Each single DNA strand is composed of a sequence of four bases nucleotides (adenine (A), guanine (G), thymine (T), cytosine (C) ) and it is paired up with a complementary strand to form a double helix\cite{dna}. Finding reversible codes is an  essential requirement in order to find codes suitable for DNA computing. Also, reversible codes may be used in certain data storage applications \cite{mas}. For example, one just needs to read the stored data from either end of a block code. In this paper, we present a general method to construct reversible $\ell$-block MT, QT and QC codes, and a number of good reversible QC codes over $GF(2)$ are obtained using this method. 

The codes presented in Table 3 are reversible and self orthogonal QC with the same parameters as BKLCs over $GF(2)$. Hence they are more desirable than BKLCs that do not have these properties. In checking the reversibility of these codes, we  used the conditions for reversibility of QC codes given in \cite{qcReversible}. 
Theorem 1 in section IV of  \cite{qcReversible} presents three conditions for reversibility of QC codes with generator matrices of the form $[g(x), g(x) \cdot f(x)]$, where $g(x) | x^m-1$ and $\gcd(f(x), h(x))) = 1$. After finding all nonequivalent generators $g(x)$ for a given block length $m$, we check the first two conditions on $g(x)$ and $h(x)$ before proceeding. If these two checks pass, then we  generate random polynomials $f(x)$ that satisfy the third condition in the theorem, resulting in a reversible QC code. Additionally, since these QC codes are constructed over $GF(2)$ they have the added property of being self-orthogonal. Theorem 1 also provides conditions for self-orthogonality of QC codes with such generator matrices and in the binary case these conditions are the same as those for reversibility. For a polynomial $f(x)$, it is customary to denote its reciprocal by $f(x)^{*}$.
\begin{theorem}
Let $g_i|x^{n_i}-a_i$ and $\gcd(h_i,f_i)=1$ for $i=1,2,...,\ell$, where $h_i=\frac{x^{n_i}-a_i}{g_i}$. If $\deg(g_if_i)=\deg(g_{\ell+1-i}f_{\ell+1-i})$, $n_i=n_{\ell+1-i}$ and $g_if_i=(g_{\ell+1-i}f_{\ell+1-i})^*$ for $i=1,2,..,\ell$, then the MT code $C$ generated by $(g_1f_1,\dots,g_{\ell}f_{\ell})$ is reversible.
\end{theorem}
\begin{proof}
Let $g_if_i=a_{i,0}+a_{i,1}x +\cdots +a_{i,k_i}x^{k_i}$, so the generator matrix of this MT code with $\ell$ blocks is given by
\begin{align*}
    G=\begin{pmatrix}
    G_1&G_2&...&G_{l-1}&G_l
    \end{pmatrix},
\end{align*}
where $G_i$ is 
\begin{align*}
G_i=
\begin{pmatrix}
a_{i,0} &a_{i,1} &a_{i,2}&...&a_{i,k_i}&0&0&...\\
0 &a_{i,0} &a_{i,1}&...&a_{i,k-1}&a_{i,k_i}&0&...\\
\vdots &\ddots & \ddots & \ddots &\ddots&\ddots&\ddots &\ddots\\
0 & 0&... &a_{i,0}&a_{i,1}&a_{i,2}&...&a_{i,k_i}\\
\end{pmatrix}_{(n_i-k_i) \times n_i}.
\end{align*}
Hence, the reversed first row of $G$ is given by 
\begin{align*}
    (&0,...,0,a_{\ell,k_\ell},a_{\ell,k_\ell-1},...,a_{\ell,1},a_{\ell,0};\\&0,...,0,a_{\ell-1,k_{\ell-1}},a_{\ell-1,k_{\ell-1}-1},...,a_{\ell-1,1},a_{\ell-1,0};\\&\vdots\\&0,...,0,a_{1,k_1},a_{1,k_1-1},...,a_{1,1},a_{1,0}).
\end{align*}
Note that $g_if_i=(g_{\ell+1-i}f_{\ell+1-i})^*$ for $i=1,2,..,\ell$, so $a_{i,j}=a_{\ell+1-i,k-j+1}$ for $i=1,2,..,\ell$ and $j=1,2,...,k$. Also, since the degree of all $\deg(g_if_i)=\deg(g_{\ell+1-i}f_{\ell+1-i})$, it is not hard to see that the reversed first row is the $(n-k)^{th}$ row, and in general the reversed $j^{th}$ row is the same as $(n-k+1-j)^{th}$ row. Therefore, the reverse of each row in $G$ is in $G$, so $C$ is reversible.
\end{proof}

As mentioned previously all of these codes were constructed using the conditions for a $2-$QC code to be reversible found in \cite{qcReversible}. This paper gives conditions on the generator polynomials $g(x)$ and $f(x) \cdot g(x) \mod(x^m-1)$ with $\gcd(h(x), f(x)) = 1$ so the first  polynomial presented in this table is $g(x)$, and the second one is  $f(x) \cdot g(x) \mod(x^m-1)$.

\begin{small}
\begin{singlespacing}
\label{tab:3}       % Give a unique label
% For LaTeX tables use
\begin{longtable}{p{2cm} p{8cm}}
\caption*{Table 3: Reversible and Self-Orthogonal Binary QC Codes with best-known parameters}
\\
\hline\noalign{\smallskip}
$[n,k,d]$ & Polynomials  \\
\noalign{\smallskip}\hline\noalign{\smallskip}

$[60,28,12]$ & [00101100111001101100101001101], [111] \\
 $[60,29,12]$ & [000111000101011000100100011111], [11] \\
 $[72,29,16]$ & [011111101110101111101101100111101011], [10111101] \\
 $[72,30,16]$ & [0110111100000111111100010010111001], [1101011] \\
 $[80,37,16]$ & [1010011110001011110000000111110001110001], [1111] \\
$[80,38,16]$ & [0011111110101010100100101101110101011], [101] \\
 $[84,24,24]$ & [111100010111100111000111101110000000100101], [1010100000000010101] \\
 $[84,25,24]$ & [001111100001011100110011110100001101000111], [110011111111110011] \\
 $[84,26,24]$ & [001101010011000011101100000100001110001101], [11100000000000111] \\
 $[84,27,24]$ & [001001101000110101101001110110000011000001], [1011111111111101] \\
 $[90,28,24]$ & [011111010111100100001011100011000011110000001], [110111011110111011] \\
 $[90,34,20]$ & [00010111011111110110000000011010111111101101], [110001100011] \\
 $[96,36,20]$ & [10010100100000001100000001110111001001100111001], [1010001000101] \\
 $[96,37,20]$ & [01100110011010101011011110010011101110010010011], [100110011001] \\
 $[96, 38, 20]$ & [110011110011101110010010111100100111011010101], [11101110111] \\
 $[100,41,20]$ & [00110100011011100110001010100100101010000111001001], [1111111111] \\
 $[102,34,24]$ & [10011111000011100010010001100001100001101110101], [101110111111011101] \\
 $[102,42,20]$ & [1010110010011100000100111101101100110001000000101], [1001111001] \\
$[104,36,24]$ & [1010111000101110001001010010110111001001101110110101], [11110000000001111] \\
$[108,46,20]$ & [110000001010000001101111000101010010011011101000001], [111111111] \\
$[110,40,24]$ & [110111111111101000011111110110011000010001011111000111], [1111100000011111]
\end{longtable}
\end{singlespacing}
\end{small}

\section{Quantum Codes}
In comparison to classical information theory, the field of quantum information theory is relatively young and there is still a lot of room for further study. The idea of quantum error correcting codes was first introduced in \cite{Quantumoriginal1} and \cite{Quantumoriginal2}.  A method of constructing quantum error correcting codes (QECC) was given in \cite{Quantumoriginal3}. Since then researchers have investigated various methods of using classical error correcting codes to construct new QECCs.
%Some researchers have studied constructing QECCs from additive constacyclic codes in the paper "New quantum codes from additive constacyclic codes" (currently under review) but 
The majority of the methods have been based on the CSS construction given  in \cite{Quantumoriginal3}. In this method, self-dual, self-orthogonal and dual-containing linear codes are used to construct quantum codes. The CSS construction requires two linear codes $C_1$ and $C_2$ such that $C_2^{\perp}\subseteq C_1$. Hence, if $C_1$ is a self-dual code, then we can construct a CSS quantum code using $C_1$ alone since $C_1^{\perp}\subseteq C_1$. If  $C_1$ is self-orthogonal, then we can construct a CSS quantum code with $C_1^{\perp}$ and $C_1$ since $C_1^{\perp}\supseteq C_1$. Similarly in the case $C_1$ is a dual-containing code. The last few tables in this section give some good QECCs obtained from both the CSS method and the method of additive codes. A number of these codes have the same parameters as those found in the literature but we obtained them using more direct and simpler constructions. Additionally, we have found other codes whose parameters  do not appear in the literature. 

The following codes are constructed using the method of quantum codes constructed from additive codes \cite{Quantumoriginal3}. After generating all divisors $g(x)$ of $x^m - a$ that generate non equivalent constacyclic codes  and finding a generator matrix of the form $[C_{g(x) \cdot f_{1}(x)}, C_{g(x) \cdot f_{2}(x)}]$ where $\gcd(f_{i}(x), h(x))) = 1$ and $C_g(x)$ represents the circulant matrix from $g(x)$, we construct an additive code from this generator matrix and check it for symplectic self orthogonality. If this check passes, then we  construct a quantum code from this additive code and compute its minimum distance. Then we can compare the minimum distance of these codes against the comparable  best known QECCs given in the database codetables.de. 

\vspace{2mm}

In Table 4, all of the codes  are optimal quantum QT codes constructed over $GF(2^2)$. Note that $\alpha$ is a root of the irreducible polynomial $x^2 + x + 1$ over $GF(2)$. 
\begin{small}
\begin{singlespacing}

\label{tab:4}       % Give a unique label
% For LaTeX tables use
\begin{longtable}{ p{1.8cm} p{9.8cm}}
\caption*{Table 4 : Optimal quantum codes over GF($2^2$) from QT Codes}
\\
\hline\noalign{\smallskip}
$[n,k,d]$ & Polynomials  \\
\noalign{\smallskip}\hline\noalign{\smallskip} 
$[[54,47,2]]$ & $[00\alpha^21110\alpha000\alpha^21110\alpha000\alpha^21110\alpha]$, \\ & $[\alpha^2\alpha\alpha^201\alpha^2\alpha0\alpha\alpha^2\alpha\alpha^201\alpha^2\alpha0\alpha\alpha^2\alpha\alpha^201\alpha^2\alpha0\alpha]$ \\
$[[56,50,2]]$ & $[10000\alpha^2\alpha\alpha1\alpha^21\alpha^2\alpha1\alpha^2\alpha\alpha\alpha\alpha100\alpha^21\alpha^210\alpha^2]$, \\ & $[\alpha^210011\alpha\alpha\alpha^21\alpha0\alpha^2\alpha^20\alpha\alpha^2\alpha^2\alpha\alpha110\alpha1\alpha^2]$ \\
$[[70,64,2]]$ & $[10\alpha^2\alpha010\alpha0\alpha\alpha1\alpha\alpha\alpha001\alpha\alpha\alpha^210\alpha0\alpha1\alpha^200\alpha^2110\alpha]$, \\ & $[\alpha^201\alpha^2\alpha\alpha\alpha^2\alpha\alpha^2101\alpha^2\alpha^201\alpha^2\alpha^20\alpha^2\alpha\alpha1011\alpha0\alpha^2\alpha^2\alpha^21\alpha0\alpha]$ \\
$[[88,80,2]]$ & $[100\alpha^21\alpha\alpha101\alpha^2\alpha\alpha\alpha^21\alpha01\alpha^2\alpha^2\alpha^20011\alpha0\alpha^2\alpha^2010\alpha\alpha^2\alpha^2\alpha0\alpha^210\alpha\alpha\alpha1]$, \\ & $[0\alpha^211\alpha00\alpha^2000100\alpha\alpha^2\alpha^2101\alpha^21\alpha1\alpha^2\alpha^20\alpha\alpha1\alpha\alpha\alpha\alpha^2\alpha\alpha011\alpha^2\alpha\alpha^21\alpha^2]$ \\
$[[98,91,2]]$ & $[11\alpha^2101011\alpha^2101011\alpha^2101011\alpha^2101011\alpha^2101011\alpha^2101011\alpha^2101]$, \\ & $[\alpha\alpha000\alpha^20\alpha\alpha000\alpha^20\alpha\alpha000\alpha^20\alpha\alpha000\alpha^20\alpha\alpha000\alpha^20\alpha\alpha000\alpha^20\alpha\alpha000\alpha^2]$ \\
\noalign{\smallskip}\hline
\end{longtable}
\end{singlespacing}
\end{small}

%\hl{Examples of QECCs over larger alphabets. Two types: i) better codes than what is known in the literature ii) ties where the construction in the literature is through an extension ring.}

We also found some QECCs whose parameters  do not appear in the literature. Hence, we consider them new quantum codes. The inspiration for this idea came from \cite{preprint}. These codes are obtained using direct sums of two cyclic codes, with generators  $g_1(x), g_2(x)$ that are divisors of  $x^m - 1$. Each $g_{i}(x)$ generates  a dual containing cyclic code $C_i=\langle g_i(x)\rangle$  which implies that their direct sum $C_1\times C_2$ is also dual containing. So, the CSS method can be used with the direct sum code and its dual. The resultant QECC will have parameters $[[2 \cdot m, 2\cdot k - 2 \cdot m, d]]_q$. Here, $k$ is defined to be the dimension of the direct sum code which is equal to $k_1 + k_2$, the sum of the dimensions of the individual cyclic codes, and $d$ is the minimum of the minimum distances of the cyclic codes generated by $g_1(x)$ and $g_2(x)$. Table 5 contains these new codes.

%\hl{Add the two new quantum codes in text around here.}
\label{tab:5}       % Give a unique label
% For LaTeX tables use
\begin{longtable}{p{2cm} p{2cm} }
\caption*{Table 5 : New QECCs}
\\
\hline\noalign{\smallskip}
$[[n,k,d]]_q$ & Polynomials  \\
\noalign{\smallskip}\hline\noalign{\smallskip}
$[[60,54,2]]_3$ & $[21], [201] $ \\
$[[72,68,2]]_3$ & $[21], [21] $ \\
$[[84,80,2]]_3$ & $[21], [21] $ \\
$[[100,96,2]]_5$ & $[41], [41] $ \\

\noalign{\smallskip}\hline
\end{longtable}

Next, we have a number of codes that have the same parameters as the best known QECCs presented in the literature. Our codes have the advantage that they are obtained in a more direct and simple construction. They were all obtained using the same method as the codes in the previous table. Codes with the same parameters listed in Table 6 were found using more complex and indirect methods detailed in \cite{QuantumTables1}, \cite{QuantumTables2}, \cite{QuantumTables3} , \cite{QuantumTables4}, \cite{QuantumTables5}, \cite{QuantumTables6}. These constructions usually involve considering an extension ring $S$ of a ground ring or field $R$, constructing a code over the extension ring, then coming back to $R$ using some sort of Gray map. We argue that it is more desirable to construct codes with the same parameters using a more direct approach. 

%\hl{For each code, give the citation}

\label{tab:6}       % Give a unique label
% For LaTeX tables use
\begin{longtable}{p{2cm} p{2cm} p{1cm}}
\caption*{Table 6 : QECC ties}
\\
\hline\noalign{\smallskip}
$[[n,k,d]]_q$ & Polynomials & Source  \\
\noalign{\smallskip}\hline\noalign{\smallskip}
$[[72,66,2]]_3$ & $[21], [111]$ & \cite{QuantumTables3}\\
$[[84,78,2]]_3$ & $[21], [201]$ & \cite{QuantumTables1} \\
$[[40,36,2]_5$ & $[41], [41] $ & \cite{QuantumTables5} \\
$[[40,34,2]]_5$ & $[41], [131] $ & \cite{QuantumTables1} \\
$[[40,28,3]]_5$ & $[3011], [3011] $ & \cite{QuantumTables5} \\
$[[60,54,2]]_5$ & $[41], [131] $ & \cite{QuantumTables3} \\
$[[90,84,2]]_5$ & $[41], [131] $ & \cite{QuantumTables4}\\
$[[100,94,2]]_5$ & $[41], [401] $ & \cite{QuantumTables1} \\
$[[140,134,2]]_5$ & $[41], [401] $ & \cite{QuantumTables6} \\

\noalign{\smallskip}\hline
\end{longtable}

\subsection{New Quantum Codes from Constacyclic Codes and a Result About Binomials}
%%% Is this refering to the last table?

%\hl{Starting the last part?}

In this section, we obtain new QECCs that are constructed by the method of CSS construction from constacyclic codes over finite fields. These codes are new in the sense that either there do not exist codes with these parameters in the literature (to the best of our knowledge) or the minimum distances of our codes are higher than the codes that are presented in the literature. Our work on this also led to a couple of theoretical results about the binomials of the form $x^n-a$ over $\Fq$. The first one (Theorem 5.3 below)  is similar to Theorem 4.1 in \cite{mt}. We start with the results on binomials.

\begin{lemma}
Let $\alpha,\beta\in \Fq$ and $\alpha\not=\beta$, where $q$ is a power of a prime $p$. Then for all $m\in\mathbb{Z}^{+}$, $\alpha^{p^m}\not=\beta^{p^m}$.
\end{lemma}
\begin{proof}
Since $\alpha\not=\beta$, it follows that 
\begin{align*}
    (\alpha-\beta)^{p^m}&\not=0\\
    \alpha^{p^m}-\beta^{p^m}&\not=0\\
    \alpha^{p^m}&\not=\beta^{p^m}
\end{align*}
\end{proof}
\begin{theorem}
Let $q$ be a power of a prime $p$ and  let $\alpha\in \Fq$. Then for any $m\in \mathbb{Z}^{+}$, there exists a unique $\beta\in \Fq$ such that $\alpha=\beta^{p^m}$.
\end{theorem}
\begin{proof} Let $\Fq=\{a_1,a_2,...,a_q\}$. Given previous lemma, for any $\alpha,\beta\in \Fq$, where $\alpha\not=\beta$, we have $\alpha^{p^m}\not=\beta^{p^m}$. Hence, $\Fq=\{a_1,a_2,...,a_q\}=\{a_1^{p^m},a_2^{p^m},...,a_q^{p^m}\}$. Since $\alpha\in \Fq$, it follows that $\alpha\in\{a_1^{p^m},a_2^{p^m},...,a_q^{p^m}\}$. Therefore, there must be a unique $\beta\in \Fq$ such that $\alpha=\beta^{p^m}$.
\end{proof}

%\hl{Note: The following result is related to ... (one in MT paper) but not a corollary of it}

\begin{theorem} Let $n\in \mathbb{Z}^{+}$, $q$ be a prime power and $a,b\in \Fq$ such that $a\not=b$. Then $\gcd(x^n-a,x^n-b)=1$ .
\end{theorem}
\begin{proof}
Let $\delta$ be a root of $x^m-a$. Then for any $k\in \mathbb{Z}^{+}$, it follows that 
\begin{align*}
    (\delta^k)^{m}-b=(\delta^m)^k-b=a-b\not=0
\end{align*}
Hence, the polynomials $x^m-a$ and $x^m-b$ share no common roots. Therefore, $\gcd(x^n-a,x^n-b)=1$.
\end{proof}

\begin{theorem}
 Let $q$ be a power of prime $p$, and let $g(x)=x^n-a\in \Fq[x]$. Then all irreducible factors of $x^n-a$ have the same multiplicity $p^z||n$, where $z$ is the largest positive integer such that $p^z$ divides $n$.
\end{theorem}
\begin{proof}
Let $n=p^z\cdot b$, where $\gcd(b,p)=1$. Given Theorem 5.2, $x^n-a$ can be written as 
\begin{align*}
    x^n-a&=x^{p^z\cdot b}-a\\
    &=(x^{b}-a')^{p^z},
\end{align*}
where $a'\in \Fq$. Let's call $f(x)=x^{b}-a'$, and then consider  $\gcd(f(x),f'(x))$. Note that $b\not =0$ and $x^{b}-a'$ is not a multiple of $x$, so it follows that $\gcd(f(x),f'(x))=\gcd(x^{b}-a',bx^{b-1})=1$. Hence, all irreducible factors of $f(x)=x^b-a'$ have the same multiplicity 1. Therefore, all irreducible factors of $g(x)=x^n-a=f^{p^z}(x)$ have the same multiplicity $p^z$.
\end{proof}
By the definition of CSS construction, we need two codes such that one is contained in the dual of the other one. By ideals inclusion we know that
\begin{align*}
    \langle g(x)f(x) \rangle \subseteq  \langle g(x) \rangle 
\end{align*}
and we consider $C_2^{\perp}=\langle g(x)f(x)\rangle \subseteq \langle g(x) \rangle =C_1 $. Here  $g(x)$ is a divisor of a binomial $x^n-a$ and it generates a constacyclic code. 

Since for a given $n\in\mathbb{Z}^{+}$ and distinct elements $\alpha,\beta\in \Fq$, $\gcd(x^n-\alpha,x^n-\beta)=1$, it follows that for any divisor $g(x)|x^n-\alpha$,  we have $g(x)f(x)\not | x^n-\beta$ for any $f(x)\in \Fq[x]$. In other words, for a given $C_1=\langle g \rangle$, where $g|x^n-\alpha$,  all possible constacyclic codes of the same length $C_2^{\perp}$ such that $C_2^{\perp}\subseteq C_1$ have the form $C_2^{\perp}=\langle gf\rangle$, where $gf|x^n-\alpha$. Hence, for given $\alpha\in \Fq$ and $n$, we only need to consider $g_i|x^n-\alpha$ and $g_if_i|x^n-\alpha$. By Theorem 5.3, any generator of $C_2^{\perp}$ will be a divisor of $x^n-\alpha$, and not $x^n-\beta$ for $\beta\not = \alpha$.
Also, based on the multiplicity of all irreducible factors of $x^n-\alpha$, for any $g(x)$, we can determine all $g(x)f(x)$ by enumerating all  possible combinations of multiplicities of irreducible factors. Therefore, our algorithm goes through every possible pair of $g(x)$ and $g(x)f(x)$, and we view each $\langle g(x)\rangle$ as $C_1$ and each $\langle g(x)f(x) \rangle$ as $C_2^{\perp}$ in the CSS construction. This way, we obtained many new QECCs. The new codes we obtained are of two types: i) those that have higher minimum distances than the codes given  in the literature, ii) those whose parameters do not appear in the literature. We present them in the following two tables.

\begin{small}
\begin{singlespacing}

\label{tab:7}       % Give a unique label
% For LaTeX tables use
\begin{longtable}{p{2.7cm}p{2.7cm} p{7.5cm} }
\caption*{Table 7 : New QECCs that have higher minimum distances than the codes given in the literature}
\\
\hline\noalign{\smallskip}
Known QECC &New $[[n,k,d]]_q$ & $g(x)$ \& $g(x)f(x)$  \\
\noalign{\smallskip}\hline\noalign{\smallskip}
$[[11,1,4]]_5$ \cite{QuantumTables11}&$[[11,1,5]]_5$ & $[124114], [1122031] $ \\
$[[30,20,3]]_5$ \cite{QuantumTables10}& $[[30,20,4]]_5$& $[142241], [14401243413033031434210441] $ \\
$[[40, 32, 2]]_5$ \cite{QuantumTables2}&$[[40, 32, 3]]_5$ & $[12342], [1203022234202444014241112314133340313] $ \\
$[[93, 75,2]]_5$ \cite{QuantumTables12}&$[[93, 75, 4]]_5$ &$[1032243334]$,$[121431401222402220013421431143342442$\\ &&$3124230120124143123112101233022341423121100111211] $ \\
$[[60, 52,2]]_5$ \cite{QuantumTables3}&$[[60, 52, 3]]_5$ &$[10434]$,$[101321202243220133133431344440014124413202$ \\ 
&&$440312032014321] $ \\
$[[78,60,3]]_3$ \cite{QuantumTables12}&$[[78,60,4]]_5$ &$[1120201112]$,$[120220102221110112120120121022021112$ \\ 
&&$0012212110210210000100012101020022] $ \\
$[[18,2,3]]_7$ \cite{QuantumTables9}&$[[18,2,6]]_7$ & $[143561634], [14631236125] $ \\
$[[12, 4, 3]]_{13}$ \cite{QuantumTables7}&$[[12, 4, 5]]_{13}$ & $[18554], [184835595] $ \\
$[[24,18,2]]_{17}$ \cite{QuantumTables7}&$[[24,18,3]]_{17}$ & $[1332], [9511521119410532108411631] $ \\
$[[21, 7, 2]]_{19}$ \cite{QuantumTables8}&$[[21, 7, 3]]_{19}$ & $[19999998], [11313131313131319999991] $ \\
\noalign{\smallskip}\hline
\end{longtable}
\end{singlespacing}
\end{small}
The final table shows the set of codes we have found whose parameters do not appear in the literature.
%\hlc[orange]{do you want me to add more? I have tons of them} Save them for the polycyclic codes paper 

\begin{small}
\begin{singlespacing}
\label{tab:8}       % Give a unique label
% For LaTeX tables use
\begin{longtable}{p{2cm} p{10.8cm} }
\caption*{Table 8 : New QECCs}
\\
\hline\noalign{\smallskip}
$[[n,k,d]]_q$ & Polynomials  \\
\noalign{\smallskip}\hline\noalign{\smallskip}
$[[61, 41, 5]]_3$ & $[10011111001], [1100101022020121010222110210221110202120101102020022] $ \\
$[[52, 7, 6]]_3$ & $[12211120111], [100000202100102111] $ \\
$[[52, 35, 4]]_3$ & $[110000012], [11020221221200111012111020100112002220020221] $ \\
$[[52, 15, 5]]_3$ & $[12010202], [12200202101110221121201] $ \\
$[[40, 11, 6]]_3$ & $[112100201012], [11112000221100111020001] $ \\
$[[40, 22, 5]]_3$ & $[1222100011], [12220111102211221220021121211111] $ \\
$[[30, 2, 8]]_3$ & $[101101000202101], [10211201020112201] $ \\
$[[32, 2, 7]]_3$ & $[1211021211201002], [122010112001001202] $ \\
$[[30, 10, 4]]_3$ & $[120212022101], [1211111102001211202201] $ \\
$[[30, 16, 4]]_3$ & $[11122111], [110112012020020210211011] $ \\
$[[33, 1, 10]]_3$ & $[12000021012021021], [111000222111222222] $ \\
$[[35, 4, 6]]_3$ & $[12122221121022], [101020002212012112] $ \\
$[[36, 1, 8]]_3$ & $[11120020102002111], [100110211221022002] $ \\
$[[36, 2, 8]]_3$ & $[12212100100121221], [1200122110112210021] $ \\
$[[36, 18, 4]]_3$ & $[1110220111], [1102010202122112212020102011] $ \\
$[[20 , 4, 5]]_3$ & $[112212211], [1122201022211] $ \\
$[[23, 1, 8]]_3$ & $[100202011222], [1202121101001] $ \\
$[[26, 9, 6]]_3$ & $[111122121], [112200100222020102] $ \\
$[[28, 0, 9]]_3$ & $[122221000112122], [122221000112122] $ \\
$[[28, 12, 6]]_3$ & $[111000121], [121100111020121002111] $ \\
$[[14, 2,5]]_3$ & $[1101011], [111202111] $ \\
$[[11, 1, 5]]_5$ & $[124114], [1122031] $ \\
$[[12, 4, 4]]_5$ & $[14102], [140124402] $ \\
$[[16, 1, 6]]_5$ & $[1320214], [11112222] $ \\
$[[20, 4, 6]]_5$ & $[121201212], [1231323013112] $ \\
$[[30, 20, 4]]_5$ & $[142241], [14401243413033031434210441] $ \\
$[[10,2,4]]_7$ & $[12134], [1230661] $ \\
$[[14, 6, 4]]_7$ & $[12056], [12462526421] $ \\
$[[40, 28, 4]]_7$ & $[1455541], [11152236066630031463035562341504546] $ \\
$[[12, 2, 5]]_{13}$ & $[154816], [1582412411] $ \\
\noalign{\smallskip}\hline
\end{longtable}
\end{singlespacing}
\end{small}

\end{document}